\documentclass[runningheads]{llncs}
\usepackage{graphicx}

\usepackage{amsmath}
\usepackage{amssymb}
\usepackage{hyperref}
\usepackage{tikz}

\newcommand{\commentSabine}[1]{}

\newcommand{\WA}{\mathrm{WA}}
\newcommand{\Min}{\mathrm{Min}}
\newcommand{\Max}{\mathrm{Max}}
\newcommand{\Prop}{\mathtt{Prop}}
\newcommand{\Sources}{\mathcal{S}}
\newcommand{\AggStrat}{\mathtt{Agg}}
\newcommand{\AggS}{\AggStrat_\Sources}
\newcommand{\leqt}{\leq_t}
\newcommand{\leqk}{\leq_k}

\newcommand{\FDE}{\textsf{BD\ }}
\newcommand{\LFDE}{\mathcal{L}_\textsf{BD}}

\newcommand{\BD}{\textsf{BD}}
\newcommand{\LBD}{\mathcal{L}_\textsf{BD}}

\newcommand{\f}{\varphi}
\newcommand{\p}{\psi}

\newcommand{\Luk}{{\mathchoice{\mbox{\rm\L}}{\mbox{\rm\L}}{\mbox{\rm\scriptsize\L}}{\mbox{\rm\tiny\L}}}}

\begin{document}
\title{Belief based on inconsistent information\thanks{The research of Marta B\'ilkov\'a was supported by the grant GA17-04630S of the Czech Science Foundation. The research of Sabine Frittella and Sajad Nazari was funded by the grant ANR JCJC 2019, project PRELAP (ANR-19-CE48-0006). The research of Ondrej Majer was supported by the grant GA16-15621S.}}
\author{Marta B\'ilkov\'a\inst{1}\orcidID{0000-0002-3490-2083} \and
Sabine Frittella\inst{2}\orcidID{0000-0003-4736-8614} \and
Ondrej Majer\inst{3}\orcidID{0000-0002-7243-1622}
\and
Sajad Nazari\inst{2}\orcidID{0000-0002-4295-2435}
%\orcidID{}
}
\authorrunning{B\'ilkov\'a, Frittella, Majer, Nazari}
\institute{Czech Academy of Sciences, Institute of Computer Science, Prague\\
\email{bilkova@cs.cas.cz}
\and
INSA Centre Val de Loire, Univ.\ Orl\'{e}ans, LIFO EA 4022, France\\
\email{sabine.frittella@insa-cvl.fr,sajad.nazari@insa-cvl.fr}
\and
Czech Academy of Sciences, Institute of Philosophy, 
Prague\\
\email{majer@flu.cas.cz}
}
\maketitle              
% typeset the header of the contribution
%%%%%%%%%%%%%%%%%%
\begin{abstract}
A recent line of research has developed around logics of belief based on evidence \cite{BBOS:WoLLIC2016,bilkovaetal2016}. The approach of \cite{bilkovaetal2016} understands belief as based on information confirmed by a reliable source.
We propose a finer analysis of how belief can be based on information, where the confirmation comes from multiple possibly conflicting sources and is of a probabilistic nature. 
We use Belnap-Dunn logic and its probabilistic extensions to account for potentially contradictory information on which belief is grounded. We combine it with an extension of \L ukasiewicz logic, or a bilattice logic, within a two-layer modal logical framework to account for belief. 
\end{abstract}
%%%%%%%%%%%%%%%%%%
\section{Introduction}
To form beliefs about the world, we collect and process data of different origins to provide us with reliable information concerning particular issues. Information derived from data typically is of a probabilistic nature, and,
%(for instance obtained by aggregating data about a population). 
as obtained from multiple sources of different origins, it inevitably is incomplete and often conflicting concerning the issues we wish to resolve. 
In this context, we propose logics to formalize how an agent can build beliefs based on information (uncertain, incomplete, and sometimes inconsistent) provided from the available collected data, and how to reason with and about such beliefs.

\emph{Incompleteness} of information alone is ever-present when reasoning about data. 
Applications, such as relational databases, often use many-valued logics to properly account for indefiniteness. Namely, Kleene’s three-valued logic \cite{kleene52} became the design choice of SQL and similar systems (the use of Kleene’s logic in this context was first proposed by \cite{Codd75}, and argued optimal in \cite{CGLibkinKR18}.) In \cite{Belnap19}, Belnap introduced a four-valued logic with intended database applications (see e.g. \cite{GMO2015}), which extends Kleene's logic, but also allows to model reasoning with non-trivial \emph{inconsistencies}. Further developed by Dunn \cite{dunn76}, Belnap-Dunn four-valued logic $\BD$, also referred to as First Degree Entailment, became a prominent logical framework which encompasses reasoning with both incomplete and inconsistent information.
This logic evaluates formulas to Belnap-Dunn square -- a lattice built over an extended set of truth values $\{t,f,b,n\}$ (true, false, both, neither), where $b$ and $n$ correspond to inconsistent and incomplete information respectively (Figure \ref*{fig:square}, middle). One of the underlying ideas of this logic is that not only truth, but also amount of information that formulas carry (reflected by the four semantical values) matters. This idea was generalized by introducing the algebraic notion of bilattices by Ginsberg \cite{ginsberg88} in the context of AI, and studied further in \cite{R2010PhD,JR2012}. Bilattices contain two lattice orders simultaneously: a truth order, and a knowledge (or an information) order. Belnap-Dunn square, the smallest interlaced bilattice, can be seen as the product bilattice of the two-element lattice (Figure \ref*{fig:square}, left) where the truth-values are pairs of classical values which can be naturally interpreted as representing two independent dimensions of information – positive and negative one\footnote{This  independence assumption has in fact a support in scientific practice – if an experiment confirming a hypothesis fails, does not automatically mean it is rejected.}. We can understand them as providing positive and negative support for statements independently. It was used to provide the logic with the double-valuation frame semantics by Dunn \cite{dunn76}.

The problem of dealing with inconsistency concerns probabilistic information as well. There are essentially two ways out. One way is to get rid of inconsistencies, the other way is to develop systems with inference rules which can work with inconsistent premises. While on the logic side there are systems providing both kinds of solutions, for example belief revision or paraconsistent logics, the majority of solutions on the probability side go for the first solution – getting rid of inconsistency (cf. the Dempster–Shafer theory of belief functions \cite{dempster1968}) – and the attempts of the second kind emerged only relatively recently. Zhou \cite{Zhou13} extends the theory of belief functions to the setting of distributive lattices, in particular bilattices and de Morgan lattices, and provides a complete logic to reason about belief functions based on \BD. Michael Dunn \cite{dunn2010} defines a probablistic framework over four-valued logic and studies properties of the resulting probabilistic entailment. The idea of an independent account for positive and negative information, underlying the double-valuation semantics of \BD, naturally ge\-ne\-ra\-li\-zes to probabilistic extensions of Belnap-Dunn four-valued logic proposed in \cite{kleinetal2020}, which we will use in this paper. It generalizes Belnap–Dunn logic in a similar way as classical probability theory generalizes propositional logic, and is referred to as theory of non-standard probabilities. 
 
When it comes to management of \emph{uncertainty}, probability and other measures of uncertainty can be understood as graded notions, as one tries to quantify the plausibility of unverified events typically over the interval [0,1]. Graded notions are one of the subjects traditionally studied by methods of fuzzy logics. As probability is not truth-functional, it does not admit a straightforward treatment by logical methods. However, one may deal with probability as a modal operator in logical systems (cf. \cite{Hamblin59}), for example in the systems of modal fuzzy logic \cite{hajek1988}. There are two main approaches to probabilistic modalities over classical logic: two-layered and intensional. 
 The two-layered logical formalism introduced in \cite{faginHalpernMegido1990,Halpern05} separates the non-modal lower language of events from the modal upper language of probabilities. The system divides into three parts: lower level of classical propositional reasoning, reasoning about probabilities consisting of the axioms that characterize probability measures in finite spaces, and the upper level of reasoning about (Boolean combinations of) linear inequalities.
H\'ajek \cite{hajek1988} proposed to replace the quantitative reasoning in form of linear inequalities
with many-valued reasoning, namely \L ukasiewicz logic, in such formalism on the upper level to obtain a fuzzy probability logic for formal reasoning under uncertainty. The graded modality “probably”, which can be used to model belief of an agent understood as a kind of subjective probability, is interpreted as a finitely-additive probability on a Boolean algebra of events with values in the real unit interval. Consequently, a class of modal logics for dealing with virtually any uncertainty measure has been covered by the formalism in \cite{CintulaN14}. In this paper, we aim at extending the framework to encompass reasoning with inconsistent probability information. 

\medskip\par\noindent
We look at an agent who considers a set of issues, has access to (multiple) sources providing positive and negative information on the issues in form of non-standard probabilities, and builds beliefs based on information aggregated from these sources.
From plethora of possible scenarios we single out two case studies that we use to illustrate the different concepts at stake. 
\begin{example}[Aggregating heterogeneous data]
\label{case:cs}
A company launching a new car model needs to decide its selling price and its advertising strategy. Hence, its
data analysts must study the reports on the sells of the previous products launched by the company and the success or failure of the advertisement campaigns. 
This study relies on factual information such as ``during the year 2015, the company sold $n$ items of product $Y$", 
but also on statement based on statistical analyses such as ``the advertisement broadcasted in June 2016 increased the sells of product $Y$ among the 20-30 years old of 30\%''.
The second statement is based on aggregated information about the buyers that might be partial and partly false. Plus, the company has access to statistical studies about the population on increase or decrease in expenses for cars. 
\end{example}
\begin{example}
[How to lead an investigation]
\label{case:human}
An investigator needs to know if one of the suspects was present at the crime scene. She
collects information from various sources: CCTV camera recordings, ATM logs, witnesses' statements, etc. No information of this kind is absolutely precise, and typically different sources of information contradict each other. 
Sources provide information of a probabilistic nature: camera recordings are imprecise due to light conditions, witnesses are not absolutely sure what they have seen. 
Moreover, the pieces of evidence confirming investigator's hypothesis that the suspect was present at the place of crime (that is, the positive information) are different from, and somewhat independent of, those rejecting it (that is, the negative information): there is a CCTV camera closed to the crime scene vs. ATM in a supermarket in a different city. 
For example, a lack of evidence supporting the suspect was present at the crime scene does not yield a proof she was not there. 
In the end, the investigator has to aggregate the available information and form beliefs about what likely happened.
\end{example}
In many scenarios we can adapt \emph{aggregation strategies} that have been introduced on classical probabilities: a company that has access to a huge amount of heterogeneous data from various sources and 
uses software capable of analyzing these data.
In this case it makes sense to consider aggregation methods that require a substantial computational power.
A natural strategy here is to evaluate sources with respect to their reliability and aggregate them by taking their weighted average.
Another kind of agent is an investigator of a criminal case who builds her opinion on the guilt of a suspect based on different pieces of evidence. We assume that all the sources are equally reliable and the investigator is very cautious and does not want to draw conclusions hastily. Hence, she relies on statements as little as all her sources agree
on them, and the aggregation she uses returns the minimum of the positive and the minimum of the negative probabilities provided by the sources. If on the other hand the investigator considers all the sources being perfectly reliable, she accepts every piece of evidence and builds her belief using the aggregation maximazing both probabilities.

In what follows, we will propose two-layer modal logics of belief of a single agent, belief that is grounded on probabilistic information  provided (positive and negative information independently) by multiple sources. The underlying logic of facts or events is chosen to be \BD, the upper logic varies between \BD\ and logics derived from \L ukasiewicz logic and based on product or bilattice algebras, to systematically account for positive and negative information independently (and thus incompleteness and conflict) on both levels.
%%%%%%%%%%%%%%%%%%%%%%%%%%%%%%%%%%%%%%%%%%%%
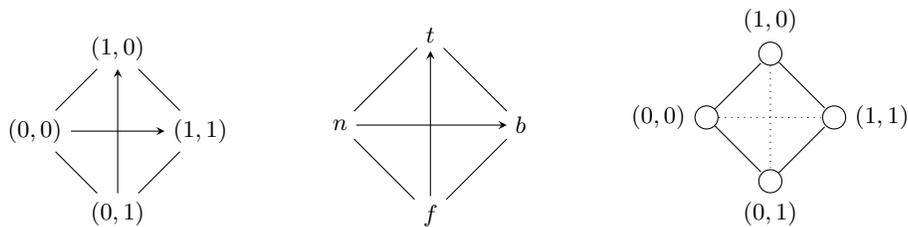
\begin{figure}
	\begin{tikzpicture}[>=stealth,relative]
	\node (U1) at (0,-1.1) {$(0,1)$};
	\node (U2) at (-1.1,0) {$(0,0)$};
	\node (U3) at (1.1,0) {$(1,1)$};
	\node (U4) at (0,1.1) {$(1,0) $};
	
	\path[-,draw] (U1) to (U2);
	\path[-,draw] (U1) to (U3);
	\path[-,draw] (U2) to (U4);
	\path[-,draw] (U3) to (U4);
	\path[->,draw] (U1) to (U4);
	\path[->,draw] (U2) to (U3);
	
	\end{tikzpicture}
	\hfill
	\begin{tikzpicture}[>=stealth,relative]

	\node (U1) at (0,-1.2) {$f$};
	\node (U2) at (-1.2,0) {$n$};
	\node (U3) at (1.2,0) {$b$};
	\node (U4) at (0,1.2) {$t$};
	
	\path[-,draw] (U1) to (U2);
	\path[-,draw] (U1) to (U3);
	\path[-,draw] (U2) to (U4);
	\path[-,draw] (U3) to (U4);
	\path[->,draw] (U1) to (U4);
	\path[->,draw] (U2) to (U3);
	\end{tikzpicture}
	\hfill
	\begin{tikzpicture}[-,>=stealth,shorten >=0.5pt,auto,node distance=1.2cm,thin,
	main node/.style={circle,draw,font=\sffamily\normalsize},%main node2/.style={circle,fill=gray!50,draw,font=\sffamily\normalsize}
	]
	
	\node[main node][label=left:{$(0,0)$}] (1) {};
	\node[main node][label={$(1,0)$}] (2) [above right of=1] {};
	\node[main node][label=below:{$(0,1)$}] (3) [below right of=1] {};
	\node[main node][label=right:{$(1,1)$}] (4) [above right of=3] {};
	
	\path[every node/.style={font=\sffamily\small}]
	(1) edge (2)
	edge (3)
	(2) edge (4)
	%edge (3)[dotted]
	(3) edge (4);
	
	\path[dotted]
	(2) edge (3)
	(1) edge (4);
	
	\end{tikzpicture}
	\caption{The product bilattice $2\odot 2 $ (left), which is isomorphic to  Dunn-Belnap square $\mathbf{4}$ (middle), and its continuous pro\-ba\-bilistic extension (right). Negation flips the values along the horizontal line.
	} 
	\label{fig:square}
\end{figure}
%%%%%%%%%%%%%%%%%%%%%%%%%%%%%%%%%%%%%%%%%
\vspace*{-1cm}
\section{Preliminaries}\label{sec:prelim}
We will first introduce algebraic structures involved as algebras of truth-values in the resulting two-layer logics of belief presented in section \ref{sec:two:layer}, where we also motivate their choice (namely bilattices of Examples \ref{ex:four}, \ref{ex:probablisticsquare}, and \ref{ex:Luksqarebilattice}, and the product algebra of Example \ref{ex:standardLuk}). Then we briefly describe the Belnap-Dunn logic, and explain the approach to probability based on Belnap-Dunn logic.
%%%%%%%%%%%%%%%%%%%%%%%%%%%%%%%%%%%%%%%%%%%%%%%%%%
\subsection{Some bilattices and MV algebras}\label{ssec:bilattices}
A \emph{bilattice} is an algebra
$\mathbf{B} = (B, \wedge, \vee, \sqcap, \sqcup, \neg)$
such that the reducts $(B, \wedge, \vee)$ and $(B,  \sqcap, \sqcup)$ are both lattices and the negation $\neg$ is a unary operation satisfying that for every $a,b\in B$,
\begin{align*}
\text{if } a \leqt b \text{ then } \neg b \leqt \neg a,
\qquad
\text{if } a \leqk b \text{ then } \neg a \leqk \neg b,
\qquad
a = \neg \neg a,
\end{align*}
with $\leqt$ (resp.\ $\leqk$) the order on $(B, \wedge, \vee)$ (resp.\ $(B,  \sqcap, \sqcup)$) called the truth (resp.\ knowledge or information) order.
A bilattice is \emph{interlaced} if each one of the four operations $\wedge$, $\vee$, $\sqcap$, $\sqcup$ is monotone w.r.t.\ both orders $\leqt$ and $\leqk$.
Bilattices, as well as interlaced bilattices, form a va\-rie\-ty. 

Given an arbitrary lattice
$\mathbf{L}=(L,\wedge_L,\vee_L)$, 
we can construct the \emph{product bilattice} 
$\mathbf{L} \odot \mathbf{L}=(L \times L , \wedge, \vee, \sqcap, \sqcup, \neg)$ as follows: for all $(a_1, a_2), (b_1,b_2) \in L \times L$, 
\begin{align*}
(a_1,a_2) \leqt(b_1,b_2)\ &\mbox{ iff } a_1\leq b_1\ \mbox{and}\ b_2\leq a_1\\
 \neg (a_1,a_2) &:= (a_2,a_1) \\
(a_1,a_2) \wedge (b_1,b_2) &:= (a_1 \wedge_L b_1, a_2 \vee_L b_2)
\\ 
(a_1,a_2) \vee (b_1,b_2) &:= (a_1 \vee_L b_1, a_2 \wedge_L b_2)
\\
(a_1,a_2) \sqcap (b_1,b_2) &:= (a_1 \wedge_L b_1, a_2 \wedge_L b_2)
\\ 
(a_1,a_2) \sqcup (b_1,b_2) &:= (a_1 \vee_L b_1, a_2 \vee_L b_2)
\end{align*}
$\mathbf{L} \odot \mathbf{L}$  is always an interlaced bilattice, and any interlaced bilattice can be represented as a product billatice:
a bilattice $\mathbf{B}$ is interlaced if and only if there is a lattice $\mathbf{L}$ such that $\mathbf{B} \cong \mathbf{L} \odot \mathbf{L}$ \cite{Avron96}.
\begin{example}
\label{ex:four}
The smallest interlaced bilattice is the product bilattice of the two-element lattice $\mathbf{2}\odot \mathbf{2}$ (Figure~\ref{fig:square} left). It is isomorphic to Dunn-Belnap square $\mathbf{4}$ used as a matrix of truth values for Belnap-Dunn logic  (Figure~\ref{fig:square} middle), with $\{t,b\}$ being the designated values. %The lattice reduct of the bilattice $\mathbf{4}$ w.r.t.\ the truth order is also called , and
% in particular it is a distributive bounded lattice (if we include the constants $t$ and $f$ in the signature). 
%Moreover it is a de Morgan algebra, and it generates the class of de Morgan algebras as a quasivariety.
\end{example}
\begin{example}
\label{ex:probablisticsquare}
A probabilistic extension of Dunn-Belnap square (Figure~\ref{fig:square} right) can be seen as based on the pro\-duct bilattice $\mathbf{L}_{[0,1]} \odot \mathbf{L}_{[0,1]}$, where 
$\mathbf{L}_{[0,1]}= ([0,1], \min, \max).$
\end{example}
A \emph{residuated lattice} is an algebra 
$\mathbb{L}=(L, \wedge_L, \vee_L, \cdot, \setminus, / )$, where the reduct
$(L, \wedge_L, \vee_L)$ is a lattice, 
$(L, \cdot)$ is a semi-group (i.e.\ the operation $\cdot$ is associative) 
and the residuation properties hold: for all $a,b,c \in L$:
\begin{center}
\begin{tabular}{c c c}
   $a \cdot b \leq c$  & \quad iff  \quad
$b \leq a \setminus c$  & \quad iff \quad
$a \leq c / b.$
\end{tabular}
\end{center}
\begin{example}
\label{ex:standardLuk}

\begin{enumerate}
    \item 
$[0,1]_{\Luk}= ([0,1], \wedge, \vee, \&_{\Luk}, \rightarrow_{\Luk})$, the standard algebra of \L uka\-siewicz logic, is a residuated lattice, an MV algebra\footnote{For more on \L ukasiewicz logic and MV algebras (in particular finite standard completeness w.r.t. $[0,1]_{\Luk}$) see e.g. \cite{FuzzyHB2}.}, 
and it ge\-ne\-ra\-tes the variety of MV algebras. 
(As $\&_{\Luk}$ is commutative, the two implications coincide.) 
For all $a,b\in [0,1]$, we define a negation ${\sim_{\Luk}}a := a \rightarrow_{\Luk} 0 := 1 - a$,
and the standard operations
\begin{align*}
a \wedge b &:= \min \{ a,b\}, &
a \&_{\Luk} b &:= \max \{0, a+b-1 \} ,
\\
a \vee b &:= \max \{ a,b\}, &
a \rightarrow_{\Luk} b &:= \min \{ 1, 1 - a + b) \}.
\end{align*}
    \item 
$[0,1]^{op}_{\Luk}= ([0,1]^{op}, \vee, \wedge, \oplus_{\Luk}, \ominus_{\Luk})$ arises turning the standard algebra upside down, and is isomorphic to the original one. Here, we have ${\sim_{\Luk}}a := 1 \ominus_{\Luk} a$,
\begin{center}
\begin{tabular}{c c}
  $a \oplus_{\Luk} b \approx {\sim}a \rightarrow_{\Luk} b = \min \{1, a + b\}$ \  & \ $a \ominus_{\Luk} b \approx {\sim}(a \rightarrow_{\Luk} b) = \max \{0, a - b\}.$
\end{tabular}
\end{center}
    \item 
Finally, we will consider the product MV algebra $[0,1]_{\Luk} \times [0,1]^{op}_{\Luk} = ([0,1] \times [0,1]^{op}, \wedge, \vee, \&, \rightarrow)$ with operations defined pointwise, ${\sim}(a_1,a_2) := a\rightarrow (0,1) = ({\sim}a_1,{\sim}a_2)$, and:
\begin{align*}
(a_1,a_2) \& (b_1,b_2) &:= (a_1 \&_{\Luk} b_1, a_2 \oplus_{\Luk} b_2)\\
(a_1,a_2) \rightarrow (b_1,b_2) &:= (a_1 \rightarrow_{\Luk} b_1, b_2 \ominus_{\Luk} a_2) = (a_1 \rightarrow_{\Luk} b_1, {\sim}a_2 \&_{\Luk} b_2).
\end{align*}
As both the projections are surjective homomorphisms of MV algebras, this algebra also generates the variety of MV algebras. 
\end{enumerate}
\end{example}
Given a residuated lattice
$\mathbf{L}=(L, \wedge_L, \vee_L, \cdot, \setminus, / )$
the \emph{pro\-duct residuated bilattice} \cite{JR2012} 
$\mathbf{L} \odot \mathbf{L} =(L \times L, \wedge, \vee, \sqcap, \sqcup, \supset, \subset, \neg )$ is defined as follows:
the reduct $(L \times L, \wedge, \vee, \sqcap, \sqcup)$
is the product bilattice $(L, \wedge_L, \vee_L) \odot (L, \wedge_L, \vee_L)$ and, for all $(a_1, a_2), (b_1,b_2) \in L \times L$,
\begin{align*}
(a_1, a_2) \supset (b_1,b_2) & := (a_1 \setminus b_1, b_2 \cdot a_1),
&
(a_1, a_2) \subset (b_1,b_2) & := (a_1 / b_1, b_1 \cdot a_2).
\end{align*}
One can then define the following operations: for all $a,b \in L \times L$,
\begin{align*}
a \rightarrow b := (a \supset b) \wedge (\neg a  \subset \neg b), && a \leftarrow b := \neg a \rightarrow \neg b,
&&
a \ast b :=  \neg (b \rightarrow \neg a).
\end{align*}
For any product residuated bilattice, 
the structure 
$(L \times L, \wedge, \vee, \ast, \rightarrow, \leftarrow , \neg )$ is a residuated bilattice endowed with an involutive negation. If $\cdot$ is commutative (associative), so is $\ast$.
\begin{example}
\label{ex:Luksqarebilattice}
The product residuated bilattice arising from the standard MV algebra is the structure $[0,1]_{\Luk} \odot [0,1]_{\Luk} = ([0,1]\times [0,1], \wedge, \vee, \sqcap, \sqcup, \supset, \neg, (0,0))$ where:
\begin{align*}
(a_1,a_2) \ast (b_1,b_2) &:= (a_1 \&_{\Luk} b_1, (a_1 \rightarrow_{\Luk} b_2) \wedge (b_1 \rightarrow_{\Luk} a_2))\\
(a_1,a_2) \rightarrow (b_1,b_2) &:= ((a_1 \rightarrow_{\Luk} b_1) \wedge (b_2 \rightarrow_{\Luk} a_2), a_1 \&_{\Luk} b_2),
\end{align*}
and $(1,1)$ acts as the unit of the $\ast$: $(1,1)\ast a = a \ast (1,1) = a$.\footnote{\label{ft:imp:star}Definitions of $\rightarrow,\ast$ match those used in \cite{CDK06} for interval based fuzzy logics, via a transformation given by $(x_1,x_2)\mapsto(x_1,1-x_2)$ (symmetry across the $(0,0.5)(1,0.5)$ line).} We define 
\begin{center}
\begin{tabular}{c l}
  ${\sim}a:= (a\supset (0,0))\sqcup \neg(\neg a\supset (0,0))$ & $= ({\sim_{\Luk}}a_1,{\sim_{\Luk}}a_2)$  \\
 $a\oplus b := ({\sim}a\supset b)\sqcup\neg({\sim}\neg a\supset\neg b)$ & $=(a_1\oplus_{\Luk}b_1,a_2\oplus_{\Luk}b_2)$ \\
 $a\ominus b := {\sim}(a\supset b)\sqcap\neg{\sim}(\neg a\supset\neg b)$ & $=(a_1\ominus_{\Luk}b_1,a_2\ominus_{\Luk}b_2)$.
\end{tabular}
\end{center}
From \cite{JR2012}, we know that the (isomorphic copies of) product residuated bilattices obtained from MV algebras form a variety, and its axiomatization can be obtained by translating the one of MV algebras (in the language of residuated lattices)\footnote{\label{ft:axiom}\cite{JR2012} hints at the correspondence between subvarieties of residuated lattices and residuated bilattices being categorial. This would mean that the mentioned variety is in fact generated by the product bilattice of the standard MV algebra. One could then use the translation from \cite{JR2012} to obtain axiomatics of the logic introduced at the end of Subsection \ref{ssec:probabilisticbelief}.}. 
\end{example}
%%%%%%%%%%%%%%%%%%%%%%%%%%%%%%%%%%%%%%%%%
\subsection{Belnap-Dunn logic}\label{ssec:FDE}
Belnap-Dunn four-valued logic $\BD$, in the propositional language $\LBD$ built from a (countable) set $\Prop$ of propositional variables using connectives $\{\wedge, \vee, \neg\}$, evaluates formulas to Belnap-Dunn square -- the (de Morgan) lattice $\mathbf{4}$ built over an extended set of truth values $\{t,f,b,n\}$ (Figure \ref{fig:square}, middle). 

The consequence relation of logic $\BD$ is given, based on the logical matrix $(\mathbf{4},F)$ with $F = \{t,b\}$ being the designated values, as 
$$
\Gamma\vDash_{\BD}\f \ \mbox{iff} \ \forall e\ (e[\Gamma]\subseteq F \to e(\f)\in F).
$$
A frame semantics can also be given for $\BD$, in two ways. Belnap-Dunn \emph{four-valued model} is a tuple $\langle W,\mathbf{4}, e\rangle$ where $W$ is a set of states and $e$ is a valuation of atomic formulas  $e: \Prop \times W \to \mathbf{4}$. The valuation is extended to formulas of $\LBD$ using the algebraic operations on $\mathbf{4}$ in the expected way.

Following Dunn's approach \cite{dunn76}, we adopt a \emph{double va\-lua\-tion model} $M = \langle W, \Vdash^+, \Vdash^- \rangle$, giving the positive and negative support of atomic formulas in the states, extending in the following way: 
\begin{center}
	\begin{tabular}{lclclcl}
		$s\Vdash^+\f \vee \psi$ &\text{ iff }&$ s\Vdash^+\f\text{\ or\ } s\Vdash^+ \psi$,
		& $\qquad$
		& $s\Vdash^+\f \wedge \psi$ &\text{ iff }&$ s\Vdash^+\f\text{\ and\ } s\Vdash^+ \psi$,
		\\
		$s\Vdash^-\f \vee \psi$ &\text{ iff }&$ s\Vdash^-\f\text{\ and\ } s\Vdash^- \psi$,
		&&
		$s\Vdash^-\f \wedge \psi $ &\text{ iff }&$ s\Vdash^-\f \text{\ or\ } s\Vdash^-\psi$.\\	
		$s\Vdash^+\neg \f$ &\text{ iff }&$s\Vdash^- \f$
		& $\qquad$
		& $s\Vdash^-\neg \f$ &\text{ iff }&$ s\Vdash^+\f$
	\end{tabular}
\end{center}
It can be seen as locally evaluating formulas in the product bilattice $2\odot 2$ (Figure \ref{fig:square} left), and thus in $\mathbf{4}$ (Figure \ref{fig:square}, middle), connecting it with the four-valued frame semantics above. 
\BD\ is completely axiomatized using the following axioms and rules:
\begin{center}
	\begin{tabular}{llll}
	$\f\wedge\p\vdash\f\ $ & $\f\wedge\p\vdash\p\ $
	&
	$\f\vdash\p\vee\f\ $ & $\f\vdash\f\vee\p\ $\\
	$\f\vdash\neg\neg \f\ $ & $\neg\neg\f\vdash\f\ $
	&
	\multicolumn{2}{l}{$\f\wedge (\p\vee \chi) \vdash (\f\wedge\p)\vee(\f\wedge\chi)$}
	\end{tabular}

	$\dfrac{\f \vdash\p, \p\vdash\chi}{ \f \vdash \chi}$\ \
	$\dfrac{ \f \vdash\p, \f\vdash\chi}{ \f \vdash\p \wedge\chi}$\ \ 
	$\dfrac{\f \vdash\chi, \p\vdash\chi} {\f \vee\p \vdash \chi}$\ \ 
	$\dfrac{\f \vdash\p} {\neg\p \vdash \neg\f}$	
\end{center}
It is known to be (strongly) complete w.r.t. the algebraic and the double va\-lua\-tion (or 4-valued) frame semantics. \BD\ is also known to be \emph{locally finite}.\footnote{It means there are only finitely many (up to inter-derivability) formulas in a fixed finite set of propositional variables. It affects the completeness of the logic in Subsection~\ref{ssec:probabilisticbelief}. More on \BD\ and its properties can be found e.g. in the thesis \cite{P2018PhD}.}

We will use \BD\ in the two-layer formalism mainly to capture the underlying information states on which sources of probabilistic information are based.
%%%%%%%%%%%%%%%%%%%%%%%%%%%%%%%%%%%%%
\subsection{Non-standard probabilities}
\label{ssec:NSProb}
The idea of independence of positive and negative information naturally ge\-ne\-ra\-li\-zes to probabilistic extensions of \FDE logic as follows. A probabilistic Belnap-Dunn ($\BD$) model \cite{kleinetal2020} is a double valuation \FDE model extended with a classical probabi\-lity measure on the power set of states $P(W)$ generated by a mass function on the set of states $W$.\footnote{The probability of a set $X\subseteq W$ is defined as the sum of masses of its elements.} The positive and negative probabi\-lities of a formula are defined as (classical) measures of its positive and negative extensions: 
$$
p^+(\f) := \sum_{s\, \Vdash^+\f} \ m(s) \ \ \ \ \ \ \ \text{and} \ \ \ \ \ \ \ p^-(\f) := \sum_{s\, \Vdash^-\f} \ m(s).
$$
The probabilities satisfy the following axioms (see \cite[Lemma 1]{kleinetal2020}): 
\begin{center}
				\begin{tabular}{ll}
			(A1) normalization & $0 \le p (\f) \le 1$ \\
			(A2) monotonicity  & if $\f \vdash_\BD \p$ 
			then $p^+ (\f) \le p^+ (\p)$ and $p^- (\p) \le p^- (\f)$
			\\
			(A3) import-export  & $p (\f \wedge \p) + p (\f \vee \p) = p (\f) + p (\p)$.
		\end{tabular}
\end{center}
These axioms are weaker than classical Kolmogorovian ones. In particular, axiom A3 can be derived from Kolmogorovian axioms of additivity and normalization, but additivity is strictly stronger and cannot be derived from A1-A3 \footnote{Considering just the inequality $p (\f \vee \p) \geq p (\f) + p (\p) - p (\f \wedge \p)$ in place of A3, we obtain belief functions on (finite) distributive lattices \cite{Halpern05,Zhou13}.}. 
As $p(\neg\f) \neq 1- p(\f)$ in general, this account of probability admits positive probability of classical contradictions and thus allows for a non-trivial treatment of classically inconsistent information. 
When defined as above, the positive and negative probabilities are mutually definable via negation as $p^-(\f) = p^+(\neg\f)$.
%\begin{equation*}
%    \mbox{(A4)}\  p^-(\f) = p^+(\neg\f)
%\end{equation*}
It has been shown in \cite[Theorem 4]{kleinetal2020} that any non-standard probability assignment (i.e., positive and negative probability satisfying the four axioms) arises from a classical probability measure on a \BD\ double-valuation model as described above.

We can diagrammatically represent non-standard probabilities on a continuous extension of Belnap-Dunn square (Figure \ref{fig:square}, right), which we can see  as a pro\-duct bilattice $\mathbf{L}_{[0,1]} \odot \mathbf{L}_{[0,1]}$. 
For example, the point $(0,0)$ corresponds to no information being provided (neither $\f$ nor $\neg\f$ is supported by any state with positive measure in the underlying model), while $(1,1)$ is the point of maximally conflicting information (both $\f$ and $\neg\f$ are ``certain" - supported by every state with positive measure). 
The left-hand triangle $(1,0),(0,0),(0,1)$ corresponds to the cases of incomplete information, the right-hand triangle $(1,0),(1,1),(0,1)$ corresponds to the cases of conflicting information. 
The vertical dashed line corresponds to the ``classical" probabilities when positive and negative support sum up to 1. 
The horizontal line represents situations where we have as much information supporting $\f$ as contradicting it.
\begin{example}[Consulting a panel]
The company assembles a
panel to which they ask  whether a word describes the car well or not. That is, they ask how much they agree with the statements: ``the car has property $\phi$ (e.g.\ being a family car)?'' and ``the car does not have property $\phi$?''
If  humans were classical agents, every person would answer with a probability $p$ that belongs to the vertical line of the probabilistic extension of the Belnap-Dunn square. However, experience has shown that often people don't reason classically \cite{ASV2016}.
When a person answers $(p^+(\phi), p^-(\phi) )$, if 
$p^+(\phi)+ p^-(\phi) > 1$, then she is conflicted about whether the property $\phi$ describes the car, if $p^+(\phi)+ p^-(\phi) < 1$, then there might be some uncertainty on how to judge whether the car has property $\phi$.  
\end{example}
%%%%%%[Aggregations]%%%%%%%%%
\subsection{Aggregating probabilities}
We model an agent that considers a set of topics listed by the atomic variables in $\Prop$, has access to sources giving information within the framework of non-standard probabilities (which we will call simply probabilities)
and builds beliefs based on these sources using a so-called \emph{aggregation strategy}. We focus on cases where the agent has no prior beliefs about the topics at stake. Depending on the context, the  aggregation strategy should satisfy different properties.

A \textit{source} $s$  is a probability assignment over the set of formulas $s  :  \ \LFDE \rightarrow [0,1]\times [0,1]$. In particular, we will later identify a source with a mass function on the \BD\ states of a double-valuation model.
An \textit{aggregation strategy} $\AggStrat$ is a function that takes in input a set of sources $\Sources = \{ s_i \}_{i \in I}$ and returns a map $\AggS : \  \LFDE \rightarrow [0,1]\times [0,1]$. For every  $\f \in \LFDE$, we denote $\AggS (\f)^+$ (resp.\ $\AggS (\f)^-$)  the positive (resp.\ negative) support assigned to $\f$.

$\AggStrat$ is \textit{monotone} if $\f \vdash \p$ implies $\AggS(\f) \leq \AggS(\p)$ for all $\f,\p \in \LFDE$ and for every $\Sources$. 
$\AggStrat$ is \textit{$\neg$-compatible} if $\AggS(\f)^- = \AggS(\neg\f)^+$  for every $\f \in \LFDE$ and for every $\Sources$. 
$\AggStrat$ \textit{preserves probabilities} if $\AggS$ is a probability for every $\Sources$.

Many aggregation strategies have been introduced on classical probabilities. Some of them, such as the (weighted) average, straightforwardly generalize to non-standard probabilities.
In the following, we present aggregation strategies for our two case studies.
\paragraph{Weighted average.} In Example \ref{case:cs}, 
a company has access to a huge amount of heterogeneous data from various sources and to software to analyse these data. A natural proposal is to grade every source $s_i$ with respect to its reliability $w_i$ and to take the \emph{weighted average} of the probabilities. The aggregation is then the map $\WA \ : \ \LBD \rightarrow [0,1]\times [0,1]$ such that, for every $\varphi \in \LBD$,
$$ \WA^+ (\varphi) := \frac{\sum_{1\leq i \leq n} w_i \cdot p_i^+(\varphi)}{\sum_{1\leq i \leq n} w_i}
\qquad \text{and} \qquad
\WA^- (\varphi) := \frac{\sum_{1\leq i \leq n} w_i \cdot p_i^-(\varphi)}{\sum_{1\leq i \leq n} w_i}.$$
One can easily prove that $\WA$ preserves probabilities, it is monotone, and $\neg$-compatible.
This aggregation strategy is however not feasible when modelling human reasoning.
\paragraph{A very cautious investigator.} In Example \ref{case:human}, an investigator
builds her opinion on the suspect based on different sources. We assume all the sources are equally reliable and the investigator does not want to draw conclusions hastily. Hence, she relies only on  statements all her sources agree on. The aggregation is then the map $\Min \ : \ \LBD \rightarrow [0,1]\times [0,1]$ such that
\begin{align*}
\Min (\varphi) := 
& \sqcap_{1\leq i \leq n} p_i(\varphi)
= \left(\min_{1\leq i \leq n} p_i^+(\varphi), 
\min_{1\leq i \leq n} p_i^-(\varphi)\right).
\end{align*}
\paragraph{Reasoning with trusted sources.} Staying in Example \ref{case:human}, we now assume all the sources are perfectly reliable. Hence, the investigator builds her belief on every statement supported by at least one source. The aggregation is then the map $\Max \ : \ \LBD \rightarrow [0,1]\times [0,1]$ such that 
\begin{align*}
\Max (\varphi) := 
& \sqcup_{1\leq i \leq n} p_i(\varphi)
= \left(\max_{1\leq i \leq n} p_i^+(\varphi), 
\max_{1\leq i \leq n} p_i^-(\varphi)\right).
\end{align*}
Here, one has high chances of  reaching contradiction. In a scientific analyses, if one gets experiments or information that is contradictory, there are two options. Either the information is incorrect or there is a mistake in the interpretation of the data.
Here, if the sources are 100\% reliable, reaching a contradiction state will simply indicate to our investigator that there is a flaw in her analysis of the problem and she needs to change perspective  to resolve the conflict.

The two latter aggregation strategies are monotone and $\neg$-compatible, and they in general do not preserve probabilities.
%%%%[Two-layer logics - very brief description]%%%%%%%
\section{Two-layer logics}\label{sec:two:layer}
To make a clear distinction between the level of events or facts, information on which the agent bases her beliefs, and the level of reasoning about her beliefs, we  use a \emph{two-layer} logical framework. The formalism originated with \cite{faginHalpernMegido1990,hajek1988}, and was further developed in \cite{CintulaN14,baldietal2020} into an abstract algebraic framework with a general theory of syntax, semantics and completeness (we will employ this framework to derive completeness of the logics we define).

Syntax $(\mathcal{L}_e, \mathcal{M}, \mathcal{L}_u)$ of a two layer logic $\mathcal{L}$ consists of a lower language  $\mathcal{L}_e$ of events or facts (we denote formulas of $\mathcal{L}_e$ by $\f,\psi,\ldots$), an upper language $\mathcal{L}_u$ (we denote formulas of $\mathcal{L}_u$ by $\alpha,\beta,\ldots$), and a set of unary modalities $\mathcal{M}$ which can only be applied to a non-modal formula of $\mathcal{L}_e$, forming a modal atomic formula of $\mathcal{L}_u$ (in particular, no nesting of modalities can occur). 

Semantics of a two layer logic $\mathcal{L}$ is, in the abstract approach of \cite{CintulaN14}, based on \emph{frames} of the form $F = (W, E, U, \langle \mu^{\heartsuit}\rangle_{\heartsuit\in \mathcal{M}}),$  where $W$ is a set of states, $E$ is a local algebra of evaluation of the lower language $\mathcal{L}_e$ within the states\footnote{For this paper, we always consider the lower algebras be the same for all states. But different algebras can be later used when modelling heterogeneous information.}, $U$ is an upper-level algebra interpreting the modal formulas, and for each modality its semantics is given by a map $\mu^{\heartsuit}: \prod_{s\in W} E \to U$, returning a value in the upper-level algebra for a tuple of values from the lower algebra (assigned to an argument formula within the states). We write algebras, but often we need to use matrices to evaluate formulas, i.e.\ algebras with a set of designated values. Such a frame is called $E$-based and $U$-measured. A model is a frame equipped with valuations of $\mathcal{L}_e$ in $E$ (the values of atomic modal formulas are then computed by $\mu$, and values of modal formulas are computed in $U$ in an expected way). A non-modal formula $\f$ is valid in a model iff it is assigned a designated value in $E$ by all the states, a modal formula $\alpha$ is valid in a model iff its value is designated in $U$. A consequence relation is defined via preserving validity in every model. It is of the sorted form $\Psi,\Gamma \vDash \xi$ where $\Psi\subseteq\mathcal{L}_e, \Gamma\subseteq\mathcal{L}_u,\xi\in \mathcal{L}_e\cup\mathcal{L}_u$.

The resulting logic as an axiomatic system $L = (L_e, M, L_u)$ consists of an axiomatics of the lower logic $L_e$, modal rules (i.e. rules with non-modal premises and modal conclusion) and modal axioms (modal rules with zero premises) $M$, and an axiomatics of the upper logic $L_u$. Proofs can be defined in the expected way. We can see that $\Psi,\Gamma \vdash \f$ iff $\Psi\vdash_{L_e} \f$, and $\Psi,\Gamma \vdash \alpha$ iff $\Psi_{MR},\Gamma \vdash_{L_u} \alpha$, where $\Psi_{MR}$ consists of conclusions of modal rules whose premises are derivable from $\Psi$ in $L_e$ (for more detail see \cite[Proposition 3]{baldietal2020}).
%%%%%%%%[The logics]%%%%%%%%%%
\subsection{Logic of probabilistic belief}\label{ssec:probabilisticbelief}
In scenarios like that of Example \ref{case:cs}, it is reaso\-na\-ble to represent agents beliefs as probabilities. In such two-layer logics, the bottom layer is that of events or facts, represented by \BD-information states. A source provides probabilistic information given as a mass function on the states, multiple sources are to be aggregated with an aggregation strategy preserving probabilities. The modality is that of probabilistic belief. For the upper-layer -- the logic of thus formed beliefs -- we propose two logics derived from \L ukasiewicz logic. The main reason to choose \L ukasiewicz logic as a starting point is that it can express the probability axioms, and contains a well-behaved (continuous) implication. We however also aim at a formalism that allows us to separate the positive and negative dimensions of information or support also on the level of beliefs (just like \BD\ does on the lower level). This motivates the use of product or bilattice algebras (those of Examples \ref{ex:standardLuk} and \ref{ex:Luksqarebilattice}) on the upper level.
\par\medskip\noindent
\textit{I. An extension of \L ukasiewicz logic with bilattice negation.} Consider the product of the standard algebra of \L ukasiewicz logic $[0,1]_{\Luk}= ([0,1], \wedge, \vee, \&_{\Luk}, \rightarrow_{\Luk})$ with the algebra $[0,1]^{op}_{\Luk} = ([0,1]^{op}, \vee, \wedge, \oplus_{\Luk}, \ominus_{\Luk})$, as introduced in Example \ref{ex:standardLuk}(3.), with only $(1,0)$ as the designated value. 
The logic of this product algebra (understood as the set of theorems - formulas always evaluated at $(1,0)$ - or as a consequence relation preserving the value $(1,0)$) is  \L ukasiewicz logic $\Luk$. It can be axiomatized in the (complete) language $\{\rightarrow, \sim\}$ by axioms of weakening, suffixing, commutativity of disjunction, and contraposition, and the rule of Modus Ponens (see the axioms below). To be able to operate the pairs of values as a positive and negative support of formulas, we extend the signature of the algebra with the bilattice negation $\neg(a_1,a_2)=(a_2,a_1)$, and extend the language to $\{\rightarrow, \sim, \neg\}$ (notice in particular, that $\oplus$ and $\ominus$ can be defined as in Example \ref{ex:standardLuk}). We obtain the following axioms and rules, and denote the resulting consequence relation $\vdash_{\Luk(\neg)}$:
\begin{align*}
\alpha &\rightarrow (\beta\rightarrow\alpha) & \neg\neg\alpha &\leftrightarrow \alpha\\
(\alpha\rightarrow\beta) &\rightarrow ((\beta\rightarrow\gamma)\rightarrow(\alpha\rightarrow\gamma)) & \neg{\sim}\alpha &\leftrightarrow {\sim}\neg\alpha\\
((\alpha\rightarrow\beta)\rightarrow\beta) &\rightarrow ((\beta\rightarrow\alpha)\rightarrow\alpha) & ({\sim}\neg\alpha\rightarrow{\sim}\neg\beta) &\leftrightarrow {\sim}\neg(\alpha\rightarrow\beta)\\
({\sim}\beta\rightarrow{\sim}\alpha) &\rightarrow (\alpha\rightarrow\beta) & \alpha, \alpha\rightarrow\beta / \beta &\ \ \ \alpha / {\sim}\neg\alpha
\end{align*}
The $\neg$ negations can be pushed to the atomic formulas, and we can thus consider formulas up to provable equivalence in a \emph{negation normal form (nnf)}, i.e.\ formulas built using $\{\rightarrow, \sim\}$ from \emph{literals} of the form $p, \neg p$. 
 It is easy to see, because we have $\neg{\sim}\alpha\leftrightarrow {\sim}\neg\alpha$ and $\neg (\alpha\rightarrow\beta)\leftrightarrow {\sim}{\sim}\neg (\alpha\rightarrow\beta)\leftrightarrow {\sim}({\sim}\neg\alpha\rightarrow {\sim}\neg\beta)$ provable. A procedure can be defined which turns each $\alpha$ into $\alpha^\neg$ in nnf, so that we can prove, by induction, that $({\sim}\alpha)^\neg\leftrightarrow {\sim}\alpha^\neg$ and $ (\alpha\rightarrow\beta)^\neg\leftrightarrow{\sim}({\sim}\alpha^\neg\rightarrow {\sim}\beta^\neg)$.
%(see Appendix \ref{app:NNF}). 
We  denote $\boxdot\Gamma := {\sim}\neg\Gamma\cup\Gamma$. 
\begin{lemma}\label{lem:lukproduct-luk}
For any \emph{finite} set of formulas $\Gamma,\alpha$ in a nnf,
\begin{center}
\begin{tabular}{c c}
   $\Gamma\vdash_{\Luk(\neg)}\alpha \ $  iff  \ &  for some finite $\Delta:\ \boxdot\Gamma,\Delta\vdash_{\Luk}\alpha,$
\end{tabular}
\end{center}
where $\Delta$ contains instances of $\neg$-axioms. 
\end{lemma}
\begin{proof}
The right-left direction is almost trivial: $\Luk$ is a subsystem of $\Luk(\neg)$, and all the axioms in $\Delta$ are provable in $\Luk(\neg)$, and, thanks to the ${\sim}\neg$-rule, $\Gamma\vdash_{\Luk(\neg)}\boxdot \gamma$ for each $\gamma\in\Gamma$.

For the other direction, we proceed in a few steps. First, we denote by $\vdash_{\Luk(\neg)^-}$ provability in $\Luk(\neg)$ \emph{without} the ${\sim}\neg$-rule. By routine induction on proofs (and using that ${\sim}\neg$ distributes from/to implications and negations), we can see that
\begin{center}
\begin{tabular}{c c}
   $\Gamma\vdash_{\Luk(\neg)}\alpha \ $  iff   & \  $\boxdot\Gamma\vdash_{\Luk(\neg)^-}\alpha.$
\end{tabular}
\end{center}
Then we can list all the instances of $\neg$-axioms in the proof in $\Delta$, and obtain:
\begin{center}
\begin{tabular}{c c}
   $\boxdot\Gamma\vdash_{\Luk(\neg)^-}\alpha \ $  iff   & \ $ \boxdot\Gamma,\Delta\vdash_{\Luk}\alpha.$
\end{tabular}
\end{center}
First, note that we can list in $\Delta$ all instances of $\neg$-axioms for all subformulas of $\Gamma,\alpha$ as well and still keep the Lemma valid. This will come handy in the following proof. Second, we stress that in the final proof $\boxdot\Gamma,\Delta\vdash_{\Luk}\alpha$ in $\Luk$, we still use language of $\Luk(\neg)$, where formulas starting with $\neg$ are seen from the point of view of $\Luk$ as atomic.\hfill $\Box$
\end{proof}
Lemma \ref{lem:lukproduct-luk} provides a translation of provability in $\Luk(\neg)$ to pro\-va\-bility in $\Luk$ and allows us to observe that the extension of $\Luk$ by $\neg$ is conservative. Now, using finite completeness of $\Luk$, we can see that $\Luk(\neg)$ is finitely strongly complete w.r.t.\ $[0,1]_{\Luk} \times [0,1]^{op}_{\Luk}$:
\begin{lemma}[Finite strong standard completeness of $\Luk(\neg)$]\label{lem:lukproductcompleteness} For a \emph{finite} set of formulas $\Gamma$,
\begin{center}
\begin{tabular}{c c}
   $\Gamma\vdash_{\Luk(\neg)}\alpha$\ \ iff  & \ \ $\forall  e:\mathcal{L} \to [0,1]_{\Luk} \times [0,1]^{op}_{\Luk}\ (e[\Gamma]\subseteq \{(1,0)\} \to e(\alpha) = (1,0)).$
\end{tabular}
\end{center}
\end{lemma}
\begin{proof}
The left-right direction is soundness, and consists of checking that the axioms are valid and the rules sound. We only do some cases:

First the ${\sim}\neg$-rule: assume that $e$ is given and $e(\alpha)=(1,0)$. Then $e({\sim}\neg\alpha) = {\sim}\neg (1,0) = {\sim}(0,1) = (1,0)$.

Next, for any $e$, $e({\sim}\neg(\alpha\rightarrow\beta))={\sim}\neg(e(\alpha)\rightarrow e(\beta))={\sim}\neg(e(\alpha)_1\rightarrow_{\Luk}e(\beta)_1, {\sim_{\Luk}}(e(\beta)_2\rightarrow_{\Luk} e(\alpha)_2 )) = ((e(\beta)_2\rightarrow_{\Luk} e(\alpha)_2 ), {\sim_{\Luk}}(e(\alpha)_1\rightarrow_{\Luk}e(\beta)_1))$,

and, $e({\sim}\neg\alpha\rightarrow{\sim}\neg\beta) = {\sim}\neg e(\alpha)\rightarrow {\sim}\neg e(\beta) = ({\sim_{\Luk}}e(\alpha)_2,{\sim_{\Luk}}e(\alpha)_1)\rightarrow({\sim_{\Luk}}e(\beta)_2,{\sim_{\Luk}}e(\beta)_1) = ({\sim_{\Luk}}e(\alpha)_2\rightarrow_{\Luk}{\sim_{\Luk}}e(\beta)_2,e(\alpha)_1 \&_{\Luk} {\sim_{\Luk}}e(\beta)_1) = ((e(\beta)_2\rightarrow_{\Luk} e(\alpha)_2 ), {\sim_{\Luk}}(e(\alpha)_1\rightarrow_{\Luk}e(\beta)_1)).$

Last, for any $e$, $e(\neg{\sim}\alpha) = \neg{\sim}e(\alpha) = ({\sim_{\Luk}}e(\alpha)_2,{\sim_{\Luk}}e(\alpha)_1) = {\sim}\neg e(\alpha) = e({\sim}\neg\alpha)$.

For the other direction, let us assume that $\Gamma\nvdash_{\Luk(\neg)}\alpha$. Then for some finite $\Delta$ containing instances of $\neg$-axioms (in particular those for subformulas of $\Gamma,\alpha$), we have $\boxdot\Gamma,\Delta\nvdash_{\Luk}\alpha.$ Because $\Luk$ is finitely standard complete, there is an evaluation $e:\mathrm{At} \to [0,1]_{\Luk}$ sending all formulas in $\boxdot\Gamma,\Delta$ to $1$, while $e(\alpha)< 1$. Here, $\mathrm{At}$ contains literals from $\Gamma,\alpha$ of the form $p,\neg p$, and atoms and formulas of the form $\neg\delta$ from $\boxdot\Gamma,\Delta$. We define $e': \Prop \to [0,1]_{\Luk}\times [0,1]_{\Luk}^{op}$ by $e'(p) = (e(p), e(\neg p))$. We can then prove, by routine induction, that for each formula $e'(\beta) = (e(\beta),e(\beta^\neg))$. We use the fact that $e[\Delta]\subseteq\{1\}$, and $\Delta$ contains all instances of $\neg$-axioms for all subformulas of $\Gamma,\alpha$.

We now immediately see that $e'(\alpha) < (1,0)$, because $e(\alpha) < 1$. 
To prove that indeed $e'[\Gamma]\subseteq\{(1,0)\}$, we use the fact that $e[\boxdot\Gamma]\subseteq\{1\}$: as for all $\gamma\in\Gamma$, $e({\sim}\neg\gamma) = 1$, $e(\neg\gamma) = e(\gamma^\neg) = 0$. 
For the latter, we again need to use the fact that $e[\Delta]\subseteq\{1\}$, and $\Delta$ contains all instances of $\neg$-axioms for all subformulas of $\Gamma$, as they prove, by means of $\Luk$, that $\neg\gamma\leftrightarrow\gamma^\neg$, and $e$ has to respect that. 
Now we conclude, that for all $\gamma\in\Gamma$, $e'(\gamma) = (e(\gamma),e(\gamma^\neg)) = (1,0)$. \hfill $\Box$
\end{proof}
We can now put together the two-layer syntax of the first two-layer logic to be
\begin{itemize}
    \item $\mathcal{L}_e = \{\wedge,\vee,\neg\}$ language of \BD,
    \item $\mathcal{M} = \{B\}$ a belief modality,
    \item $\mathcal{L}_u = \{\rightarrow, {\sim}, \neg\}$ language of $\Luk(\neg)$.
\end{itemize}
The intended models can be described as follows: the lower layer is a double-valuation model of \BD\ $(W,\Vdash^+,\Vdash^-)$ (a set of states $W$, and the two support relations, which in fact can be seen as arising from an evaluation of formulas of \BD\ locally in the states in the product bilattice $2\odot 2$, which is isomorphic to $\mathbf{4}$, as noted in Subsection \ref{ssec:FDE}). A source is given by a mass function on the states $m_i: W \to [0,1]$, we assume there are $n$ sources, and each source comes with a weight $w_i\in [0,1]$. For a non-modal formula $\f\in\mathcal{L}_e$, we obtain the value $||B\f||\in[0,1]_{\Luk} \times [0,1]^{op}_{\Luk}$ as a pair of its positive and negative probabilities as follows.
First, for each source $m_i$, we have $(\sum_{v\, \Vdash^+\f} m_i(v), \sum_{v\, \Vdash^-\f} m_i(v)) = (p_i^+(\f),p_i^-(\f)).$ Now, applying the weighted average aggregation strategy we obtain
$$||B\f|| = \left( \frac{\sum_{1\leq i \leq n} w_i \cdot p_i^+(\f)}{\sum_{1\leq i \leq n} w_i},\frac{\sum_{1\leq i \leq n} w_i \cdot p_i^-(\f)}{\sum_{1\leq i \leq n} w_i}\right).$$
%%%%%%%%%%%%%%%%%
The modal part $M$ of the two-layer logic consists of two axioms and a rule reflecting directly the axioms of probabilities
listed in Subsection \ref{ssec:NSProb}:\footnote{Considering just the right-left implication in the first axiom, we can express belief functions.}
\begin{center}
\begin{tabular}{l c}
   $B(\f\vee\p)\leftrightarrow(B\f \ominus B(\f\wedge\psi))\oplus B\psi$ & \ \  $B\neg\f\leftrightarrow\neg B\f$\\
$\f\vdash_{\BD}\psi / \vdash_{\Luk(\neg)} B\f \rightarrow B\psi$ &
\end{tabular}
\end{center}
The resulting logic is $(\BD, M, \Luk(\neg))$. As $\BD$ is locally finite and strongly complete w.r.t.\ $\mathbf{4}$, and $\Luk(\neg)$ is finitely strongly complete w.r.t.\ $[0,1]_{\Luk} \times [0,1]^{op}_{\Luk}$, we can apply \cite[Theorems 1 and 2]{CintulaN14} directly to obtain finite strong completeness (soundness of the modal axioms and rules is easy to see). But first, we need to observe that the frames as we have described them can be seen within the framework of \cite{CintulaN14}:

The frames, seen in the format of \cite{CintulaN14}, are $F = (W, \mathbf{4}, [0,1]_{\Luk}\times [0,1]^{op}_{\Luk}, \mu^B)$, formulas of $\mathcal{L}_e$ are evaluated locally in the states of $W$ using $\mathbf{4}$, as in the four-valued models for \BD\ (which we can see as equivalent to the double-valuation models). The interpretation of modalities $\mu^B$ is computed as follows. A source is given by a mass function on the states $m: W \to [0,1]$. Each source comes with a weight $w_i\in [0,1]$. Given $\mathbf{e}\in\prod_{v\in W}\mathbf{4}$, we obtain, for each source $m_i$, first the following sums of weights over states: 
$(\sum_{\mathbf{e}_v\in\{t,b\}} m_i(v), \sum_{\mathbf{e}_v\in\{f,b\}}m_i(v)) = (p_i^+(\mathbf{e}),p_i^-(\mathbf{e}))$. The assignment $\mu^B$ now computes the weighted average of those as follows:
$$ \mu^B(\mathbf{e}) = \WA(\mathbf{e}) = \left(\frac{\sum_{1\leq i \leq n} w_i \cdot p_i^+(\mathbf{e})}{\sum_{1\leq i \leq n} w_i},\frac{\sum_{1\leq i \leq n} w_i \cdot p_i^-(\mathbf{e})}{\sum_{1\leq i \leq n} w_i}\right).$$

Thus, for a non-modal formula $\f\in\mathcal{L}_e$, applying $\mu^B$ to the tuple of its values in the states (which we denote by $||\phi||$), we obtain the value of $B\f$ as $||B\f||\in[0,1]_{\Luk} \times [0,1]^{op}_{\Luk}$ as a pair of its positive and negative probabilities as follows:
First, for each source we have\footnote{The value of $\f$ in $v$ being among $\{t,b\}$ means it is positively supported in $v$, i.e.  $v\Vdash^+\f$. Similarly $\{f,b\}$ means negative support.} 
$$
\left(\sum_{v\, \Vdash^+\f} m_i(v), \sum_{v\, \Vdash^-\f} m_i(v)\right) = (p_i^+(\f),p_i^-(\f) ).
$$
Now, applying the weighted average aggregation we obtain
$$
||B\f|| = \mu^B(||\phi||) = \WA(||\phi||) = \left(\frac{\sum_{1\leq i \leq n} w_i \cdot p_i^+(\f)}{\sum_{1\leq i \leq n} w_i},\frac{\sum_{1\leq i \leq n} w_i \cdot p_i^-(\f)}{\sum_{1\leq i \leq n} w_i} \right).
$$
We can now conclude the completeness as follows:
\begin{corollary}
$(\BD, M, \Luk(\neg))$ is \emph{finitely strongly complete} w.r.t.\ $\mathbf{4}$ based, $[0,1]_{\Luk} \times [0,1]^{op}_{\Luk}$-measured frames validating $M$.
\end{corollary}
In such frames, $\mu^B$ interprets $B$ as a probability:
For a frame to validate the axioms in $M$ means they are sent to $(1,0)$, by an evaluation in $[0,1]_{\Luk} \times [0,1]^{op}_{\Luk}$ induced by $\mu^B$ over the lower state valuations (which determines values of modal atomic formulas). An equivalence $\alpha\leftrightarrow\beta$ is evaluated at $(1,0)$ iff the values of $\alpha$ and $\beta$ are equal. $B(\f)^M = \mu^B(\f^M) = (p^+(\f),p^-(\f))$. Therefore, the first two axioms say that
\begin{align*}
p^+(\f\vee\psi) = (p^+(\f) - p^+(\f\wedge\psi))+ p^+(\psi)\ &\mbox{and}\ p^+(\neg\f)= \neg p^-(\f)\\
p^-(\f\vee\psi) = (p^-(\f) - p^-(\f\wedge\psi))+ p^-(\psi)\ &\mbox{and}\ p^-(\neg\f)= \neg p^+(\f).
\end{align*}
Similarly, the fact that the frame validates the rule say that $p^+$ ($p^-$) are monotone (antitone) w.r.t. $\f\vdash_{\BD}\psi$.
Analogous observation holds for the case the upper logic is the bilattice one.

From \cite[Theorem 4]{kleinetal2020}, we know that it is the
induced probability function of exactly one mass function on the $\BD$ canonical model, which in fact yields completeness w.r.t.\ the intended frames described above (with a single source). 

\par\medskip\noindent
\textit{II. A bilattice \L ukasiewicz logic}. Alternatively, if we wish to use full expressivity of a bilattice language, we can take in the upper layer $\mathcal{L}_u =\{\wedge,\vee,\sqcap,\sqcup,\subset,\neg,0\}$ to be the language of the \emph{product residuated bilattice} $[0,1]_{\Luk} \odot [0,1]_{\Luk} = ([0,1]\times [0,1], \wedge, \vee, \sqcap, \sqcup, \supset, \neg, (0,0))$, defined in the spirit of \cite{JR2012} in Example \ref{ex:Luksqarebilattice}. We evaluate  formulas of the upper logic in the matrix $([0,1]_{\Luk} \odot [0,1]_{\Luk},F)$ with $F = \{(1,a) \mid a\in [0,1]\}$ as the designated values, so that we send $0$ to $(0,0)$. 
The constants and connectives $\top,\bot,1,\ast,\rightarrow,{\sim},\oplus,\ominus$ are definable as follows:
\begin{align*}
{\sim}\alpha &:= (\alpha\supset 0)\sqcup\neg(\neg\alpha\supset 0) & \top &:= 0\supset 0 \ \ \bot:= \neg\top \ \ 1:= {\sim}0 \\
\alpha\rightarrow\beta &:= (\alpha\supset\beta)\wedge(\neg\beta\supset\neg\alpha) & \alpha\oplus\beta &:= ({\sim}\alpha\supset\beta)\sqcup\neg({\sim}\neg\alpha\supset\neg\beta)\\
\alpha\ast\beta &:= \neg(\beta\rightarrow\neg\alpha) & \alpha\ominus\beta &:= {\sim}(\alpha\supset\beta)\sqcup\neg{\sim}(\neg\alpha\supset\neg\beta)
\end{align*}
For an evaluation $e$, it holds that $e(\alpha\rightarrow\beta)\in F$ iff $e(\alpha\rightarrow\beta)\geq_t (1,1)$ iff $e(\alpha)\leqt e(\beta)$. The upper logic $L_u$ as a consequence relation is defined to be $$\Gamma\vDash_{\Luk\odot\Luk}\alpha\ \mbox{iff }\ \forall e (e[\Gamma]\subseteq F \to e(\alpha)\in F).$$
The intended frames now use $[0,1]_{\Luk}\odot [0,1]_{\Luk}$ as the upper algebra, otherwise semantics of atomic modal formulas is computed (from multiple sources) as in the previous logic. We also obtain literally the same modal axioms $M$ as above. Only here, apart from a very generic completeness w.r.t.\ $\mathbf{4}$ based frames, where the upper algebra is an algebra (in fact the Lindenbaum-Tarski algebra) of the upper logic, we cannot provide a better insight at the moment and leave axiomatization of $L_u$, and completeness w.r.t $[0,1]_{\Luk}\odot [0,1]_{\Luk}$-measured frames to further investigations (cf.\ footnote \ref{ft:axiom}).

\subsection{Logic of monotone coherent belief}\label{ssec:monotonebelief}
The simplest logic we propose to deal with scenarios like the one of Example \ref{case:human} is of the form $(\BD, M, \BD)$. Both lower and upper languages are the language of \BD, $\mathcal{M}$ consists of a single belief modality $B$. The intended frames are based on double-valuation semantics of \BD\ as before, only now we evaluate formulas of the upper logic in the bilattice $\mathbf{L}_{[0,1]}\odot \mathbf{L}_{[0,1]}$ on Fi\-gu\-re \ref{fig:square} (right). A source is given by a mass function on the states $m_i: W \to [0,1]$, we again assume there are $n$ sources. For a non-modal formula $\f$, we obtain the value $||B\f||\in\mathbf{L}_{[0,1]}\odot \mathbf{L}_{[0,1]}$ as follows.
First, for each source $m_i$, we have $(\sum_{v\, \Vdash^+\f} m_i(v), \sum_{v\, \Vdash^-\f} m_i(v)) = (p_i^+(\f),p_i^-(\f)).$ Now, applying the $\Min$\ aggregation strategy we obtain
$$||B\f|| = \left(\min_{1\leq i \leq n} p_i^+(\f), 
\min_{1\leq i \leq n} p_i^-(\f)\right).$$
Similarly, we may use the $\Max$\ aggregation strategy when reasoning with trusted sources. As before, we can see the frames inside the framework of \cite{CintulaN14} to derive completeness: frames are of the form $F = (W, \mathbf{4}, \mathbf{L}_{[0,1]}\odot \mathbf{L}_{[0,1]}, \mu^B)$ where  $\mu^B:  \prod_{v\in W} \mathbf{4}\to \mathbf{L}_{[0,1]}\odot \mathbf{L}_{[0,1]}$ computes the $\Min$\ ($\Max$) aggregation of the probabilities given by the individual sources. 
In general this aggregation strategy does not yield a probability, but it is monotone and $\neg$-compatible. 
This motivates considering logic $(\BD, M, \BD)$, where the modal part $M$ consists of the following two axioms and a rule
\begin{center}
\begin{tabular}{c c}
  $B\neg\f\dashv\vdash_{\BD_u}\neg B\f $  & \ \ \ \ $\f\vdash_{\BD_e}\psi / B\f \vdash_{\BD_u} B\psi$.
\end{tabular}
\end{center}
As $\BD$ is strongly complete w.r.t.\ both $\mathbf{4}$ and $\mathbf{L}_{[0,1]}\odot \mathbf{L}_{[0,1]}$\footnote{Because it has $(2\odot 2,\{(1,0),(1,1)\})$ as a sub-matrix: the obvious embedding is a \emph{strict} homomorphism of de Morgan matrices - it preserves and reflects the filters.}, 
we can apply \cite[Theorem 1]{CintulaN14} to conclude that $(\BD, M, \BD)$ is strongly complete w.r.t.\ $\mathbf{4}$-based $\mathbf{L}_{[0,1]}\odot \mathbf{L}_{[0,1]}$-measured frames validating $M$. In such frames, $\mu^B$ interprets $B$ as a monotone and $\neg$-compatible assignment (not necessarily a probability). We cannot in general see it as coming from a measure, or a set of measures\footnote{It is not hard to provide an example of such assignment which cannot be obtained by $\Min\ (\Max)$ aggregation of probabilities.}, on the lower states (to recover sources), and connect it with the intended semantics. 

One could however replace the upper language with the full bilattice language, consider modalities indexed by sources, and express the $\Min\ (\Max)$ aggregations explicitly using $\sqcap, \sqcup$ connectives.
%%%%[Further directions]%%%%
\section{Conclusion and further directions} 
We have proposed two-layer logics of belief based on potentially inconsistent probabilistic information coming from multiple sources. The framework keeps positive and negative aspect of information (support, evidence, belief) separate, though inter-linked, in both layers of the semantics, and thus allows for reasoning with inconsistencies, in contrast to getting rid of them. Doing so, we believe we have laid groundwork to a modular framework to model reasoning with inconsistent probabilistic information. 

We see our contribution in the following: 
to see how Belnap-Dunn's logic \BD\ (on the lower layer, and behind the non-standard probabilities) can be combined with many-valued reasoning on the upper layer provides a novel example of two-layer logics for reasoning under uncertainty. The only examples considered so far used either classical logic \cite{faginHalpernMegido1990}, or quantitative reasoning in form of linear inequalities on the upper layer \cite{Zhou13}. 
The logic $\Luk(\neg)$, extending \L ukasiewicz logic with bi-lattice negation, we introduced in Subsection \ref{ssec:probabilisticbelief} and proved its finite strong standard completeness, is to our best knowledge new and might be of independent interest. (The same can be said about the bi-lattice \L ukasiewicz logic, which however remains to be axiomatized and its completeness studied.)

The project is subject to ongoing work.
Apart from investigating further the logics proposed in this paper, we are pursuing the following research directions: 

In the continuation of \cite{FMPTW20} that generalises Dempster-Shafer theory \cite{shafer1976} to finite lattices, we are currently working on  adapting the theory to the \BD-based setting, and putting it in context of existing literature on belief functions. This would allow us to consider Dempster-Shafer combination rule as another aggregation strategy.

To cover cases when a source does not give an opinion about each formula of the language, we need to account for sources providing partial probability maps. Also cases where sources provide heterogeneous information need to be included.

An important direction to move further is to capture dynamics of information and belief given by updates on the level of sources, and to generalize the framework to the multi agent setting involving group modalities and dynamics of belief. Specifically, forming group belief, like common and distributed belief, will involve communication and sharing or pooling of sources. It might call for a use of various upper-layer languages, among those we see the ones with additional (nestable) modalities inside the upper logic to account for reflected, higher-order beliefs, in contrast to the beliefs grounded directly in the sources.
%%%%%%%%%%%%%%%%%%%%%%%%%%%%
\bibliographystyle{splncs04}
\bibliography{DALI_20.bib}
%%%%%%%%%%%%%%%%%%%%%%%%%%%%
%%%%%%%%%%%%%%%%%%%%%%%%%%%%
\end{document}